\documentclass[lettersize,journal]{IEEEtran}
\IEEEoverridecommandlockouts
\usepackage[dvipdf]{graphicx,color}
\usepackage{amssymb}
\usepackage{enumerate}
\usepackage{amsmath}
\usepackage{amsfonts}
\usepackage{balance}
\usepackage{multirow}
\usepackage{makecell}
\usepackage{color}
\usepackage{algorithm}
\usepackage{stfloats}
\usepackage{float}
\usepackage{algpseudocode}
\usepackage{color}
\usepackage{varwidth}
\usepackage{multicol}
\usepackage{subfigure}
\usepackage{xspace}
\usepackage{enumerate}
\usepackage{xcolor,cite,etoolbox}
\usepackage{bm}
\usepackage{amsthm}
\usepackage{booktabs}

\newtheorem{theorem}{Theorem}

\usepackage{multicol} 
\usepackage{array}
\newcolumntype{M}[1]{>{\raggedright\arraybackslash}m{#1}}

\newtheorem{corollary}{Corollary}
\newtheorem{remark}{Remark}
\usepackage[font=footnotesize,labelfont=scriptsize]{caption}
\usepackage{ragged2e}

\def\BibTeX{{\rm B\kern-.05em{\sc i\kern-.025em b}\kern-.08em
    T\kern-.1667em\lower.7ex\hbox{E}\kern-.125emX}}

\usepackage{titlesec}
\usepackage[margin=0.42in]{geometry}
\makeatletter 
\pretocmd\@bibitem{\color{black}\csname keycolor#1\endcsname}{}{\fail}
\newcommand\citecolor[1]{\@namedef{keycolor#1}{ \color{red}}}
\makeatother

\begin{document}

 {\title{\huge  Uplink Rate Splitting Multiple Access with Imperfect Channel State Information and Interference Cancellation}}	
	\author{{Farjam Karim, \textit{Graduate Student Member}, \textit{IEEE},  Nurul Huda Mahmood, Arthur S. de Sena, \textit{Member}, \textit{IEEE}, \\Deepak Kumar, \textit{Member}, \textit{IEEE}, Bruno Clerckx, \textit{Fellow}, \textit{IEEE}, 
 Matti Latva-aho, \textit{Fellow}, \textit{IEEE}}\\
\thanks{\hrulefill}
		\thanks{F. Karim, N.H. Mahmood, A. S. de Sena, D. Kumar and M. Latva-aho are with Centre for Wireless Communications, University of Oulu, Finland. (e-mail: \{farjam.karim, nurulhuda.mahmood, arthur.sena, deepak.kumar,  matti.latva-aho\}@oulu.fi.)\\
        B. Clerckx is with the Department of Electrical and Electronic Engineering, Imperial College London, London SW7 2AZ, U.K. (e-mail: b.clerckx@imperial.ac.uk)
 } 
 \thanks{This work This work was supported by 6G Flagship (Grant Number 369116) funded by the Research Council of Finland, and Business Finland's 6GBridge-6CORE Project (Grant Number: 8410/31/2022).}	}

	\maketitle

\begin{abstract}
 This article investigates the performance of uplink rate splitting multiple access (RSMA) in a two-user scenario, addressing an under-explored domain compared to its downlink counterpart. With the increasing demand for uplink communication in applications like the Internet-of-Things, it is essential to account for practical imperfections, such as inaccuracies in channel state information at the receiver (CSIR) and limitations in successive interference cancellation (SIC), to provide realistic assessments of system performance. Specifically, we derive closed-form expressions for the outage probability, throughput, and asymptotic outage behavior of uplink users, considering imperfect CSIR and SIC. We validate the accuracy of these derived expressions using Monte Carlo simulations. Our findings reveal that at low transmit power levels, imperfect CSIR significantly affects system performance more severely than SIC imperfections. However, as  the transmit power increases, the impact of imperfect CSIR diminishes, while the influence of SIC imperfections becomes more pronounced.  Moreover, we highlight the impact of the rate allocation factor on user performance.  Finally, our comparison with non-orthogonal multiple access (NOMA) highlights the outage performance trade-offs between RSMA and NOMA. RSMA proves to be more effective in managing imperfect CSIR and enhances performance through strategic message splitting, resulting in more robust communication.
	\end{abstract}
 

\section{Introduction}

The sixth generation (6G) of wireless connectivity necessitates a re-evaluation of existing communication methodologies, with a focus on achieving much more efficient and effective resource allocation. Notably, emerging multiple access techniques such as power domain non-orthogonal multiple access (NOMA), space-division multiple access, and rate-splitting multiple access (RSMA) are central to these advancements~\cite{Clerckx_Proc_24}. In contrast to orthogonal multiple access methods like time-division multiple access and frequency-division multiple access, which allocate resources individually to users~\cite{Ghosh_jan_22_access}, RSMA offers innovative approaches to resource allocation and interference management, paving the way for the high performance required in 6G networks~\cite{Clerckx_Proc_24, yijie_comm_survey_22}.

 The concept of rate-splitting as a multiple access technique was initially proposed as early as 1996~\cite{Rimoldi_TII_96}. Despite the foundational significance of this early work, research in this area remains sparse.  More recently, in~\cite{Bruno_Jsac}, the authors demonstrated RSMA's potential as a next-generation technology, addressing interference management and multi-user communication challenges. This approach surpasses traditional methods such as orthogonal multiple access, NOMA, and space-division multiple access.
Although several studies on the analytical framework of uplink NOMA exist~\cite{zou_twc_noma_2019, jamali_may_20_twc}, there are only a handful of studies exploring the analytical aspects of uplink RSMA under Rayleigh fading assuming perfect channel state information (CSI) and perfect successive interference cancellation (SIC)~\cite{Liu_clercx_2024_TWC, Tegos_letter_mar_22,Jaiwei_tcom_2024,Zeng_TVT_2019}.  In~\cite{Jaiwei_tcom_2024}, the uplink multiple-input-multiple-output (MIMO) RSMA physical layer was studied in detail, with link-level simulations confirming RSMA's performance benefits over baseline schemes like NOMA. Although focused on a general MIMO setting, the study assumes perfect CSI and SIC. In~\cite{Zeng_TVT_2019}, the authors investigated Single-Input Multiple-Output NOMA for uplink Internet-of-Things communications, analyzing a rate-splitting scheme to ensure max-min fairness under a perfect scenario. These assumptions, however, are unrealistic for real-life deployment, where imperfections in CSI estimation and SIC accuracy significantly affect system performance. Recognizing this gap, our work analyzes the performance of uplink RSMA under Nakagami-$m$ fading, explicitly accounting for channel state information at the receiver (CSIR) and  SIC imperfections, where two uplink users communicate simultaneously with an access point (AP). By analyzing the impact of these imperfections on outage probability and throughput, our study provides realistic performance benchmarks for uplink RSMA. Moreover, Nakagami-$m$ fading is particularly relevant because it generalizes several common fading distributions, such as Rayleigh and Rician, allowing for a more flexible modeling of practical wireless communication scenarios. 

In this work, we provide novel insights into uplink RSMA by analyzing its performance under both perfect and imperfect CSIR and SIC conditions. First, we derive analytical expressions for outage probability, throughput, and asymptotic outage probability, considering both perfect and imperfect CSIR and SIC.
Next, we show that CSIR imperfections have a greater impact on user performance than SIC imperfections at low transmit power levels. However, as user transmit power increases, the impact of imperfect CSIR diminishes, while the influence of SIC imperfections becomes more pronounced highlighting the critical importance of incorporating these imperfections in system design. Additionally, we emphasize the significance of the rate allocation factor for system performance in uplink RSMA.
Lastly, we compare uplink RSMA with NOMA and demonstrate that RSMA outperforms NOMA, achieving lower outage floors, thus showcasing its superior ability to manage interference in practical scenarios.
These contributions emphasize the essential role of CSIR and rate allocation factor in uplink RSMA design and highlight RSMA’s advantage over NOMA in handling real-world imperfections.
\subsection{Notations} The symbols $f_{X}(x)$ and $F_{X}(x)$ respectively denote the probability density function (PDF) and cumulative distribution function (CDF) of a random variable $X$, $m$ and $\hat{\Omega}$ are shape and severity parameter of Nakagami-$m$ fading, $\mathcal{CN}(0,\sigma^2)$ represents complex Gaussian with zero mean and $\sigma^2$ as variance, $\binom{\cdot}{\cdot}$ denotes the binomial coefficient, and $\gamma(\cdot, \cdot)$ and $\Gamma(\cdot)$ are the lower incomplete gamma function and Gamma function.


\section{system model}
\label{sec:systemModel}
Consider an RSMA-aided uplink communication network, where two single-antenna users $U_1$ and $U_2$ transmit data symbols to a single-antenna-equipped AP{\footnote{The consideration of multiple antennas at the AP and users goes beyond the scope of this work. This possibility arises as a potential research direction.}}. All devices are assumed to operate in half-duplex mode. The channel gain between the user $U_i$, with $i\in\{1,2\}$, and the AP is given by $h_i = {\hat{h}}_{i}+{\Tilde{h}}_{ie}$, where $\hat{{h}}_{i}$ is the estimated channel gain and ${\Tilde{h}}_{ie}$ represents the channel estimation errors. The estimation errors are modeled as ${\Tilde{h}}_{ie}\sim\mathcal{CN}(0, \Omega_{h_{ie}})$ and $\hat{h}_{i}$ follows follows a Nakagami-$m$ distribution.  According to~\cite{ Farjam_WCNC_dubai_24}, the channel estimation error variance is defined as $\Omega_{h_{ie}} = \Omega_{{i}}/(1+\delta\rho_{i}\Omega_{{i}})$, where $\Omega_{i}$ is the variance of $h_i$, $\rho_{i}=P_{i}/\sigma^2$ denotes the signal-to-noise ratio (SNR) of the transmitted signal, and $\delta>0$ is the quality parameter for channel estimation. Therefore, the variance of $|\hat{{h}}_{i}|$ can be calculated as $\hat{\Omega}_i = \left({\Omega}_i - \Omega_{h_{ie}}\right)$. 

Following the principles of uplink RSMA, the boundary points of the capacity region are achieved when one of the user splits its message, as demonstrated in ~\cite{Liu_clercx_2024_TWC, Jiawei_clercx_Mar_CL23}. The determination of the user splitting its message and the corresponding power allocation is managed by the AP~\cite{Liu_clercx_2024_TWC}. Without loss of generality, we assume that $U_2$ has a superior channel condition and therefore splits its message $x_2$ into two parts namely, $x_{2,1}$ and $x_{2,2}$,  allocating $\alpha_{2,1}$ and $\alpha_{2,2}$ fraction of the $U_2$ transmit power $P_2$ to $x_{2,1}$ and $x_{2,2}$, respectively, where $\alpha_{2,1}+\alpha_{2,2}=1$. On the other hand, $U_1$ transmits its message $x_1$ with transmit power $P_1$. These three different messages are then encoded into three separate streams, $s_{2,1}$, $s_{2,2}$ and $s_1$, which are transmitted simultaneously to the AP.
Thus, the signal received at the AP can be expressed as 
\begin{align}
    \label{Received_signal}
    y = &\sqrt{P_{1}}s_1\left(\hat{h}_1+{\Tilde{h}}_{1e}\right)+\sqrt{P_2}\left(\hat{h}_2+{\Tilde{h}}_{2e}\right)\nonumber\\
&\times\left(\sqrt{\alpha_{2,1}}s_{2,1}+\sqrt{\alpha_{2,2}}s_{2,2}\right)+\eta,
\end{align}
where $\eta$ denotes additive white Gaussian noise (AWGN) with $\eta\sim\mathcal{CN}\left(0, \sigma^2\right)$.  We assume, without loss of generality, that $s_{2,1}$ is the first stream to be decoded. In~\cite{Rimoldi_TII_96}, it has been observed  that the decoding order  $s_{2,1} \to s_1 \to s_{2,2}$ is optimal for achieving all boundary points of the capacity region.
Moreover, if a given stream cannot be decoded successfully, the subsequent streams are unlikely to be decoded. Therefore, if one stream fails, the decoding process terminates. Consequently, to fully decode the message from $U_2$,  the AP must decode all three received streams using SIC.  For $U_1$, the AP needs to decode $s_{2,1}$ and $s_{1}$.
Therefore, the signal-to-interference-plus-noise ratio (SINR) for $U_2$ can be expressed as 
\begin{align}\label{U2_SINR_21}
    \gamma_{2,1}\!=\! \frac{|\hat{h}_2|^2\rho_2\alpha_{2,1}}{|\hat{h}_2|^2 \alpha_{2,2}\rho_2+|\hat{h}_1|^2 \rho_1+\rho_1\Omega_{h_{1e}}+\rho_2\Omega_{h_{2e}}+1}.
\end{align}
 Once, $s_{2,1}$ has been successfully decoded, the AP attempts to decode $s_{1}$ with a SINR of 
\begin{align}\label{U1_SINR}
    \gamma_{1}\!=\!\frac{|\hat{h}_1|^2\rho_1}{\left(|\hat{h}_2|^2+\Omega_{h_{2e}}\right)\rho_2\left( \alpha_{2,2}+\alpha_{2,1}\Xi_2\right)+\rho_1\Omega_{h_{1e}}\!\!+\!\!1}.
\end{align}
where $\Xi_{i}\in (0, 1)$ denotes the level of imperfection in SIC, with $\Xi_i=0$ for perfect SIC. 
After successfully decoding $s_{2,1}$ and $s_{1}$, the AP attempts to decodes the second stream of $U_2$ with a SINR of 
    \begin{align}\label{U2_SINR_22}
     \gamma_{2,2}\!=\! \frac{|\hat{h}_2|^2\rho_2\alpha_{2,2}}{|\hat{h}_2|^2\rho_2\alpha_{2,1}\Xi_2\!\!+\!\Omega_{h_{2e}} D_1\!\!+\!\left(\!|\hat{h}_1|^2\!\!+\!\!\Omega_{h_{1e}}\!\right) \!\rho_1 \Xi_1\!+\!1},
\end{align}
 where $D_1 \!=\!{\rho_2\left(\alpha_{2,2}+\alpha_{2,1}\Xi_2\right)}$.

\section{Performance Evaluation}
In this section, we analyze the performance of uplink RSMA aided system by deriving closed-form expressions for outage probability and throughput. As mentioned earlier that $\hat{h}_i$ is Nakagami-$m $ distributed, therefore, $|\hat{h}_i|^2$ follows a Gamma distribution whose PDF and CDF are given as~\cite{Farjam_WCNC_dubai_24}
 \begin{align}\label{pdf}
     f_{\left|\hat{h}_i\right|^2}(x)=\frac{m_i^{m_i} x^{m_i-1}}{\hat{\Omega}_i^{m_i} \Gamma\left(m_i\right)} \exp \left(-\frac{m_i}{\hat{\Omega}_i} x\right), \quad x\geq 0.\end{align}

    \begin{align}  \label{CDF} F_{\left|\hat{h}_i\right|^2}(x)= \frac{1}{\Gamma\left(m_i\right)} \gamma\left(m_i, \frac{ m_i }{\hat{\Omega}_i}x\right), \quad x\geq 0.
 \end{align} 
{\textit{ Outage Probability Analysis:}} Let $R_1$ and $R_2$ be the target rates for $U_1$ and $U_2$, respectively. Since $U_2$'s message is divided, the target rates for the two resulting streams $s_{2,1}$ and $s_{2, 2}$ are $R_{2,1} = \varphi R_2$ and $R_{2,2} = \left(1-\varphi\right) R_2$, respectively, where $0\leq\varphi\leq 1$  is the rate allocation factor~\cite{Jiawei_clercx_Mar_CL23}. Thus, the outage probability threshold for $U_1$ is given by $\gamma_{1}^{th} = 2^{R_1}-1$ where as for $U_2$'s split message can be given as $\gamma_{2,1}^{th} = 2^{R_{2,1}}-1$ and $\gamma_{2,2}^{th} = 2^{R_{2,2}}-1$.
The respective outage probability is analyzed in the following theorems:
\begin{theorem} The outage probability of $U_1$ considering  imperfect CSIR and SIC can be expressed as 
  \begin{align}\label{1st_out}
   P^{\text{Out}}_{1} = P^{\text{Out}}_{2,1}+(1-P^{\text{Out}}_{2,1})P^{\text{Out}}_{11},
 \end{align}  
 where $P^{\text{Out}}_{2,1}$ and $P^{\text{Out}}_{11}$ are given in \eqref{P_New_out21} and \eqref{P_out_11}, respectively,  $T_{2,1}\!\!=\!\! \frac{{m_2} \gamma_{2,1}^{th}}{\hat{\Omega}_{2}\rho_2\left(\alpha_{2,1}-\gamma_{2,1}^{th}\alpha_{2,2}\right)}$, $A_{2,1} \!\!=\!\! \left(\rho_1\Omega_{h_{1e}}\!\!\!+\!\!\rho_2\Omega_{h_{2e}}\!\!\!+\!\!1\right) $, $C_1 \!\!=\!\! \frac{m_1}{\hat{\Omega}_1}$,  $B_{2,1}\!=\!\rho_1$,  $G_1 \!=\!\frac{{m_1} \gamma_{1}^{th}}{\hat{\Omega}_{1}\rho_1}$, $D_1 \!=\!{\rho_2\left(\alpha_{2,2}+\alpha_{2,1}\Xi_2\right)}$, $D_2 = \frac{m_2}{\hat{\Omega}_2}$, and $A_{1} \!\!=\!\! \left(\rho_1\Omega_{h_{1e}}+\rho_2\Omega_{h_{2e}}\left(\alpha_{2,2}+\alpha_{2,1}\Xi_2\right)+1\right)$.
\end{theorem}
\begin{figure*}[t!]
    \begin{align}\label{P_New_out21}
        P^{\text{Out}}_{2,1}= 1\!-\!\!\sum\limits_{p=0}^{m_2-1}\sum\limits_{q=0}^{p}\binom{p}{q}\frac{\left(T_{2,1} B_{2,1}\right)^q \left(T_{2,1} A_{2,1}\right)^{(p-q)}\exp\left(-T_{2,1} A_{2,1}\right) C_1^{m_1} \Gamma\left({m_1}+q\right)}{\Gamma(m_1)(p!) \left(C_1+T_{2,1}B_{2,1}\right)^{m_1+q}}
    \end{align}\hrulefill

    \begin{align}\label{P_out_11}
        P^{\text{Out}}_{11}= 1\!-\!\!\sum\limits_{j=0}^{m_1-1}\sum\limits_{k=0}^{j}\binom{j}{k}\frac{\left(G_{1} D_{1}\right)^k \left(G_{1} A_{1}\right)^{(j-k)}\exp\left(-G_{1} A_1\right) D_2^{m_2} \Gamma\left({m_2}+k\right)}{\Gamma(m_2)(j!) \left(D_2+G_{1}D_{1}\right)^{m_2+k}}
    \end{align}\hrulefill

    \begin{align}\label{P_new_22}
        P^{\text{Out}}_{2,2}= 1\!-\!\!\sum\limits_{p=0}^{m_2-1}\sum\limits_{q=0}^{p}\binom{p}{q}\frac{\left(T_{2,2} B_{11}\right)^q \left(T_{2,2} A_{2,2}\right)^{(p-q)}\exp\left(-T_{2,2} A_{2,2}\right) C_1^{m_1} \Gamma\left({m_1}+q\right)}{\Gamma(m_1)(p!) \left(C_1+T_{2,2}B_{11}\right)^{m_1+q}}
    \end{align}\hrulefill
\end{figure*}
\begin{proof}
    Refer to Appendix A.
\end{proof}
\begin{remark}
    The outage probability of $U_1$ considering perfect CSIR and SIC can be obtained by substituting high $\delta$ values for evaluating $\Omega_{h_{ie}}\; \forall i \in \{1,2\}$ and $\Xi_2=0$.
\end{remark}
\begin{corollary}
The throughput of $U_1$ under imperfect CSIR and SIC can be evaluated as
\begin{align}\label{throughpu1_1st}
    \mathcal{T}_1 = \left(1-P^{\text{Out}}_1\right)R_1.
\end{align}
Moreover, by substituting the $P_{1}^{\text{out}}$ values obtained through Remark 1 in \eqref{throughpu1_1st}, $\mathcal{T}_1$ for perfect CSIR and SIC case can also be obtained.
\end{corollary}

\begin{theorem} All three messages need to successfully decoded to decode user 2's message. Hence, the outage probability of $U_2$ considering  imperfect CSIR and imperfect SIC can be expressed as 
  \begin{align}\label{ou_22}
    P^{\text{Out}}_{2} &= 
   P^{\text{Out}}_{2,1}+(1-P^{\text{Out}}_{2,1})P^{\text{Out}}_{11}\nonumber\\
   &+(1-P^{\text{Out}}_{2,1})(1-P^{\text{Out}}_{11})P^{\text{Out}}_{2,2},
\end{align}
 where $P^{\text{Out}}_{2,2}$ is given in  \eqref{P_new_22}, $T_{2,2}\!\!=\!\! \frac{{m_2} \gamma_{2,2}^{th}}{\hat{\Omega}_{2}\rho_2\left(\alpha_{2,2}-\gamma_{2,2}^{th}\alpha_{2,1}\Xi_2\right)}$, $A_{2,2} = \left(\rho_1\Omega_{h_{1e}}\Xi_1 + \rho_2\Omega_{h_{2e}}\left(\alpha_{2,2}+\alpha_{2,1}\Xi_2\right)+1 \right)$ and $B_{11}=\rho_1\Xi_1$.
\end{theorem}
\begin{proof}
The proof can be obtained by following a similar approach as in Appendix A.\end{proof}
\begin{remark}
    The outage probability of $U_2$ considering perfect CSIR and perfect SIC of \eqref{U2_SINR_21} can be obtained by substituting  high $\delta$ values for evaluating $\Omega_{h_{ie}}\; \forall i \in \{1,2\}$ and $\Xi_2=0$ in \eqref{ou_22}.
\end{remark}
\begin{remark}
    The outage probability of $U_2$ considering imperfect SIC of \eqref{U2_SINR_21} and perfect SIC of \eqref{U1_SINR} with perfect CSIR can be obtained by substituting  $P^{\text{Out}}_{2,2}=\frac{1}{\Gamma(m_2)}\gamma\left(m_2, \frac{m_2{\gamma^{th}_{2,2}}}{\hat{\Omega}_{2}\rho_2\left(\alpha_{22}-{\gamma^{th}_{2,2}}\alpha_{21}\Xi_2\right)}\right)$ and high $\delta$ values for evaluating $\Omega_{h_{ie}}\; \forall i \in \{1,2\}$ in \eqref{ou_22}. Additionally, by setting $\Xi_2=0$, $P^{\text{Out}}_{2}$ for the perfect SIC case can also be obtained.
\end{remark}
\begin{figure*}[t]
    \centering
    \begin{minipage}[b]{0.48\textwidth}
        \centering
        \includegraphics[width=\textwidth]{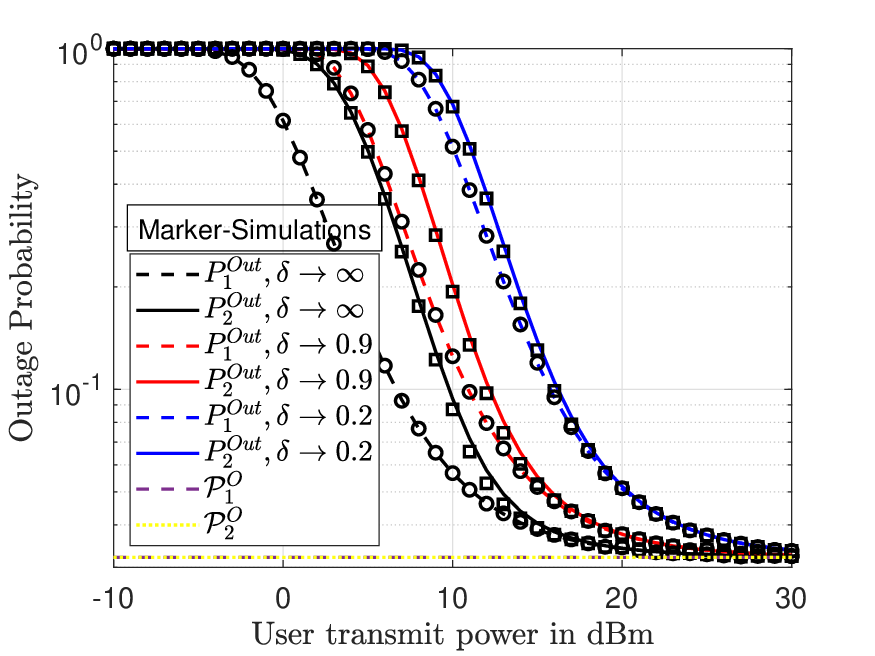}
        \caption{$\delta$ Effect on Outage vs SNR.}
        \label{fig:RSMA_ICSI}
    \end{minipage}
    \begin{minipage}[b]{0.48\textwidth}
        \centering
        \includegraphics[width=\textwidth]{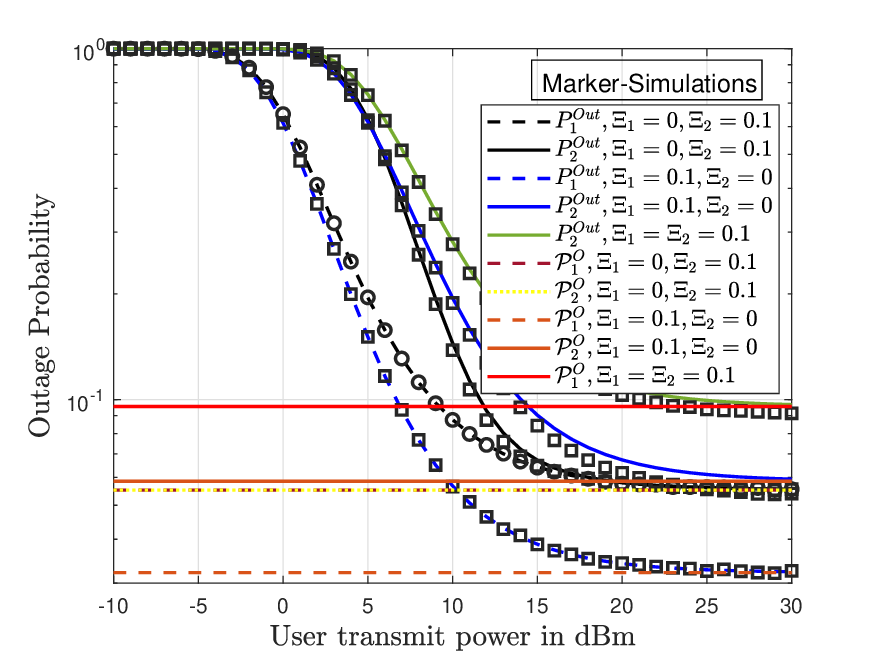}
        \caption{$\Xi_{i}$ Effect on Outage vs SNR.}
        \label{fig:RSMA_SIC}
    \end{minipage}
\end{figure*}


\begin{corollary}
The throughput of $U_2$  with imperfect CSIR and imperfect SIC can be evaluated as
\begin{align}\label{throughpu1_2nd}
    \mathcal{T}_2 = \left(1-P^{\text{Out}}_2\right)\left(R_{2,1}+R_{2,2}\right).
\end{align}
Additionally, by substituting $\Xi_2=0$ in $P_{2}^{\text{Out}}$ values obtained through Remark 3 in \eqref{throughpu1_2nd}, $\mathcal{T}_2$ for perfect CSIR and perfect SIC  can  be obtained.
\end{corollary}

\begin{theorem}
    The asymptotic outage probability of $U_1$ in the high SNR regime under imperfect SIC can be expressed as 
    \begin{align}\label{Asymp_11}
       {\mathcal{P}}_{1}^{O}\approx \mathcal{P}^O_{2,1}+(1-\mathcal{P}^O_{2,1})\mathcal{P}^O_{11},
    \end{align}
    where $\mathcal{P}^O_{2,1}=1-\sum\limits_{p=0}^{m_2-1}\frac{{C_1^{m_1}}{\left({\mathbb{T}_{2,1}}\right)^p}{\Gamma(m_1+p)}}{(p!){\Gamma(m_1)}{{\left(C_1+{\mathbb{T}_{2,1}}\right)}^{m_1+p}}}$, $\mathcal{P}^O_{11}\!\!=1-\sum\limits_{j=0}^{m_1-1}\frac{{D_2^{m_2}}{\left({\mathbb{G}_{1}}\right)^j}{\Gamma(m_2+j)}}{(j!){\Gamma(m_2)}{{\left(D_2+{\mathbb{G}_{1}}\right)}^{m_2+j}}}$,  $\mathbb{T}_{2,1}= \frac{{m_2} \gamma_{2,1}^{th}}{\hat{\Omega}_{2}\left(\alpha_{2,1}-\gamma_{2,1}^{th}\alpha_{2,2}\right)}$, and  $\mathbb{G}_{1}= \frac{{m_1} \gamma_{1}^{th}\left(\alpha_{2,2}+\alpha_{2,1}\Xi_2\right)}{\hat{\Omega}_{1}}$.
\end{theorem}
\begin{proof}
    Please refer to Appendix B.
\end{proof}
\begin{remark}
    The asymptotic outage probability of $U_1$ with perfect SIC can be obtained by substituting $\Xi_2=0$ in \eqref{Asymp_11}.
\end{remark}
\begin{theorem}
    The asymptotic outage probability of $U_2$ in the high SNR regime under imperfect SIC can be expressed as 
    \begin{align}\label{Asymp_2}
       {\mathcal{P}}_{2}^{O}&\approx \mathcal{P}^O_{2,1}+(1-\mathcal{P}^O_{2,1})\mathcal{P}^O_{11}\nonumber\\
       &+(1-\mathcal{P}^O_{2,1})(1-\mathcal{P}^O_{11})\mathcal{P}^O_{2,2},
    \end{align}
    where $\mathcal{P}^O_{2,2}=1-\sum\limits_{p=0}^{m_2-1}\frac{{C_1^{m_1}}{\left({\mathbb{T}_{2,2}}\right)^p}{\Gamma(m_1+p)}}{(p!){\Gamma(m_1)}{{\left(C_1+{\mathbb{T}_{2,2}}\right)}^{m_1+p}}}$ and $\mathbb{T}_{2,2}= \frac{{m_2} \gamma_{2,2}^{th}\Xi_1}{\hat{\Omega}_{2}\left(\alpha_{2,2}-\gamma_{2,2}^{th}\Xi_2\alpha_{2,1}\right)}$.
    \end{theorem}
    \begin{proof}
The proof can be obtained by following a similar approach as in Appendix B.\end{proof}
    \begin{remark}
    The asymptotic outage probability of $U_2$ considering  perfect SIC  can be obtained by substituting $\mathcal{P}^O_{2,2} =\frac{{\mathfrak{T}_{2,2}}^{m_2}}{\Gamma(m_2)\; m_2}$ in \eqref{Asymp_2}, where $\mathfrak{T}_{2,2}=\frac{m_2\gamma^{th}_{2,2}}{\hat{\Omega}_2\rho_2\alpha_{2,2}}$.
\end{remark}
\begin{proof}
    Please refer to Appendix C for Remark 5.
\end{proof}

 
\section{Results and Discussions}
In this section, we present Monte Carlo simulations to verify the accuracy of the derived closed-form expressions. The simulations used a distance-dependent path-loss model between $U_{i}$ and AP  given by  $\varphi/d_{i}^{\tau}$, where $i\in \{1,2\}, \varphi=1$ metre represents the reference path-loss value, $d_{i}$ is the distance between $U_i$ and AP, and $\tau=3.8$ is the path-loss exponent. 
Unless specified otherwise, all parameters for the simulations are given in Table~\ref{t3}. 
\begin{figure*}[t]
    \centering
    \begin{minipage}[b]{0.32\textwidth}
        \centering
        \includegraphics[width=\textwidth]{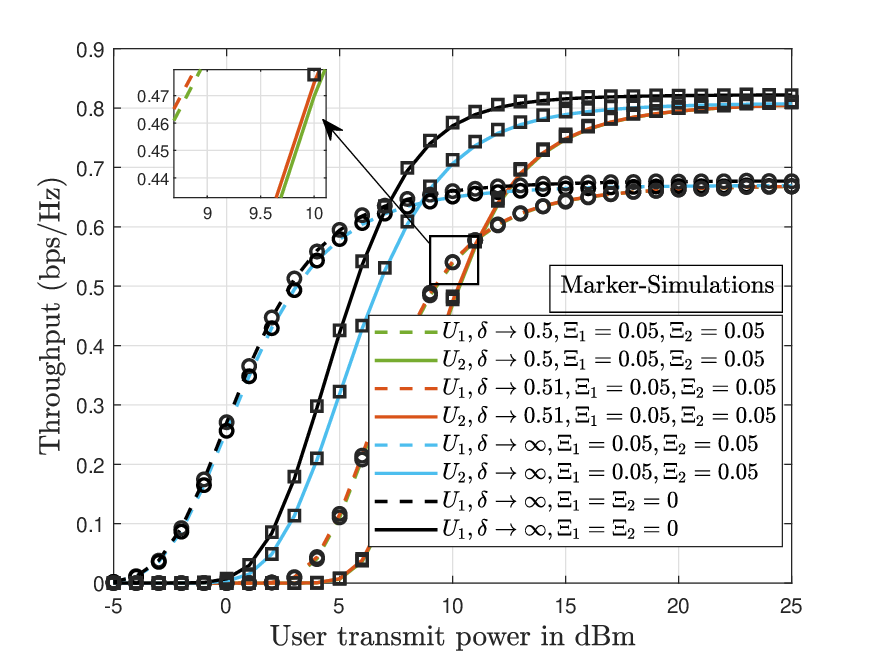}
        \caption{$\delta$ and $\Xi_{i}$ Effect on Throughput.}
        \label{fig:throughput}
    \end{minipage}
     \begin{minipage}[b]{0.32\textwidth}
        \centering
        \includegraphics[width=\textwidth]{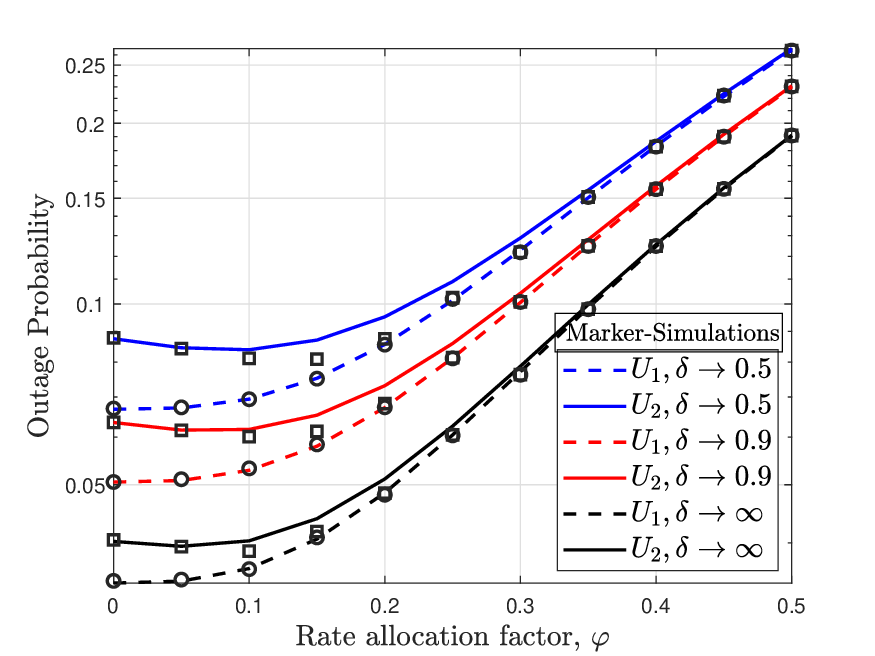}
        \caption{Outage Probability vs $\varphi$.}
        \label{fig:rate}
    \end{minipage}
    \begin{minipage}[b]{0.32\textwidth}
    \centering
    \includegraphics[width=\textwidth]{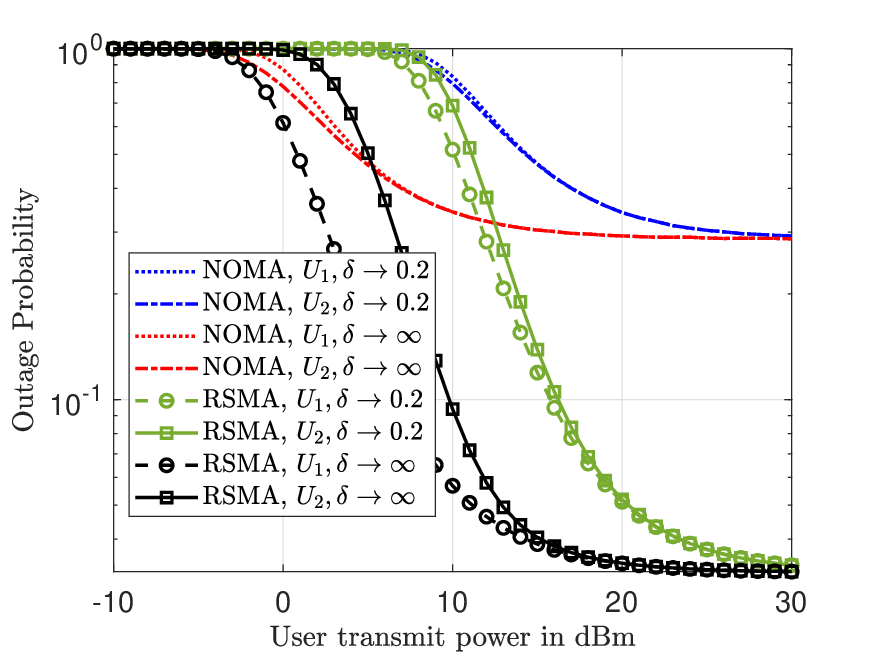}
    \caption{NOMA vs RSMA Comparison.}
    \label{fig:NOMa}
\end{minipage}\end{figure*}
    
\begin{table}[h!]	\renewcommand{\arraystretch}{1.0}
		\centering
		\caption{ Simulation Parameters}
		\label{t3}
			\resizebox{\columnwidth}{!}{\begin{tabular}{|l|l|l|l|l|l|}
			\hline
			Parameter         & Value         & Parameter & Value  & Parameter & Value  \\ \hline
			$m_1  $   &     $4 $   	&  $m_2$  &     $3$ & $\sigma^2$   &     $-100 $ dBm \\ \hline
			
			 $R_1$   &     $0.7$ &	   $ R_{2}$    &    $ 0.95 $ &  $\varphi $  &     $0.15$  	 \\ \hline
			
		 $\alpha_{2,1}$       & $0.27 $ 	&  $\alpha_{2,2} $   &     $0.73$    &      $U_1$-AP  &     $75$m  	 \\ \hline

    $U_2$-AP      & $70 $m 	&  $ $   &     $ $    &      $ $  &     $ $  	 \\ \hline
			
		
		\end{tabular}}
	\end{table}


Fig.~\ref{fig:RSMA_ICSI}  demonstrates the $U_i$ outage probability  with respect to their  $\rho_{i}$, considering the impact of channel estimation quality parameter $(\delta)$ with perfect SIC, whereas the Fig.~\ref{fig:RSMA_SIC} represents the  $P^{\text{Out}}_i$ and asymptotic outage behaviour $({\mathcal{P}}_{i}^{O})$ considering imperfect SIC factor $(\Xi_{i})$ with perfect CSIR. It is clearly evident from the figures that   $P^{\text{Out}}_1$ and $P^{\text{Out}}_2$ expressions derived in \eqref{1st_out} and \eqref{ou_22} are in close agreement with simulated results for all values of $\delta$ and $\Xi_i$. An increase in the value of $\delta$ from $0.2$ to $0.9$ can significantly reduce the outage probability for both users before approaching the floor.  This improvement is due to the enhanced accuracy in channel estimation, which provides a precise power control and interference management, thus improving the overall reliability of the communication link. 

In Fig.~\ref{fig:RSMA_SIC}, we observe how the impact of the imperfect SIC factor $\Xi_1$ is more significant on the outage probability of the second user, compared to when only $\Xi_2$ (the interference affecting both the users) is present. For instance, when $\Xi_1=0$ and $\Xi_2=0.1$, the outage probability for the first user, is approximately $8 \times 10^{-2}$, and for the second user, it is around $1.3 \times 10^{-2}$ at a transmit power of $10$ dBm. However, when the situation is reversed, with $\Xi_2=0$ and $\Xi_1=0.1$, the outage probability for the first user improves to about $5 \times 10^{-2}$, but the second user's performance degrades to approximately $1.8 \times 10^{-2}$. The reason for this trade-off between both the users is attributed towards a three-fold situation. First, when $\Xi_2=0$, $U_1$ benefits from an increase in its desired received power because there is no residual SIC present when decoding its signal. Second, since $U_2$ message is split into two parts, its total signal power is divided across two SINRs, with the first user’s signal affecting the final outcome of $U_2$ performance. Third, if $\Xi_1$ is non-zero, it directly effects the desired signal power of \eqref{U2_SINR_22}.  Furthermore, the impact of imperfect SIC significantly influences system performance; however, imperfect CSIR has a much more pronounced effect at low transmit SNR. This is because it affects the SINR across all decoding steps, as outlined in \eqref{U2_SINR_21}, \eqref{U1_SINR}, and \eqref{U2_SINR_22}, thereby influencing the overall performance of both users right from the start. In contrast, imperfect SIC impacts the  performance after the first message has been decoded. As the user transmit power increases, the effect of imperfect SIC becomes more dominant, while the impact of imperfect CSIR diminishes. Furthermore, the asymptotic outage floor behavior of $U_1$ and $U_2$ can also be observed as $\rho_1$ and $\rho_2$ increase without bound at the same rate. This behavior occurs because, at high SNR, the system's performance becomes limited by factors other than noise, such as interference and imperfections in SIC, leading to a non-zero probability of outage even as SNR increases.

Fig.~\ref{fig:throughput} depicts the throughput plotted against transmit SNR, considering different values of $\delta$ and $\Xi_i$. The derived expressions in \eqref{throughpu1_1st} and \eqref{throughpu1_2nd} closely match the simulations, as expected. Additionally, $\mathcal{T}_1$ and $\mathcal{T}_2$ begin to saturate as they approach their respective target rates for different values of $\delta$ and $\Xi_i$. Furthermore, it can be observed that the impact of imperfect SIC on $\mathcal{T}_1$ and $\mathcal{T}_2$ when $\delta \to 0.5$ and $\Xi_1 = \Xi_2 = 0.05$ is slightly overcompensated when $\delta \to 0.51$ for the same $\Xi_1 = \Xi_2$ values across the entire range of transmit SNR. This observation highlights the dominance of channel estimation quality parameter factor over the imperfect SIC factor in terms of user performance enhancement. Additionally, the throughput degradation for $U_1$ and $U_2$ at 5 dBm from $\mathcal{T}_1 \approx 5.94\times 10^{-1}$ to $5.79 \times 10^{-1}$ bits per second per Hertz (bps/Hz) and $\mathcal{T}_2 \approx 4.25 \times 10^{-1}$ to $3.22 \times 10^{-1}$ bpcu, respectively, when $\Xi_{1} = \Xi_{2} = 0$ changes to $\Xi_{1} = \Xi_{2} = 0.05$ highlights the importance of incorporating realistic assumptions in order to capture the practical system behavior while designing future wireless connectivity solutions.

Fig.~\ref{fig:rate} shows the outage behavior of $U_1$ and $U_2$ with respect to the rate allocation factor ($\varphi$) considering perfect SIC. It is interesting to observe that when $\varphi > 0.3$, $P^{\text{Out}}_1$ and $P^{\text{Out}}_2$ become extremely close to each other as $P^{\text{Out}}_1$ depends on the decoding of the $s_{2,1}$ stream. Additionally, if $\varphi = 1$, the $\gamma_{2,2}^{th}$ will effectively become zero. This observation, along with the trend in Fig.~\ref{fig:rate}, suggests that as $\varphi$ increases beyond a certain limit, which is between $0.1$ and $0.2$ in our case, the performance of both users degrades. The reason for this trend is due to the presence of $\varphi$ while designing $\gamma_{2,1}^{th}$ and $\gamma_{2,2}^{th}$. As $\varphi$ increases, the threshold designated for \eqref{U2_SINR_21} also increases, which in turn affects $P^{\text{Out}}_1$.

Fig.~\ref{fig:NOMa} compares the uplink outage probabilities for NOMA and RSMA, with target rates of 0.75 and 0.85, and $U_2$ being decoded first in NOMA transmission. The figure shows that NOMA reaches its outage floor much earlier than RSMA, underscoring its limited ability to effectively handle interference from other users. 
Additionally, the outage probabilities for $U_1$ and $U_2$ in NOMA are closely aligned, which is due to the absence of message splitting in NOMA, leading to a higher decoding threshold for $U_1$ compared to RSMA users.
RSMA consistently shows better performance for $U_1$ compared to $U_2$. However, in the NOMA scenario, $U_2$ exhibits slightly better performance than $U_1$. Finally, RSMA handles imperfect CSIR more effectively than NOMA, as evident when $\delta \to 0.2$, where RSMA achieves lower outage values than NOMA.

\section{Conclusion}
This work provides significant insights into the uplink RSMA technique by deriving analytical expressions for outage probability and throughput under both perfect and imperfect CSIR and SIC. We have also evaluated the asymptotic outage behavior, which complements our understanding of the system's performance.
Our results reveal that the imperfect CSIR has a more pronounced effect on user performance  than imperfect SIC at low SNRs. However, as transmit power increases, the impact of imperfect SIC increases while that of imperfect CSIR slowly fades away, highlighting the importance of accurate channel estimation and effective message cancellation in system design. Furthermore, we found that user performance is highly sensitive to the rate allocation factor.
Finally, our findings indicate that RSMA manages imperfect CSIR more effectively than NOMA and achieve much lower outage floors than NOMA across different SNR values. 
\appendices
 \section{}
 \begin{proof}
    In order to derive $U_1$'s outage probability, we first need to evaluate \eqref{U2_SINR_21} outage expression and then \eqref{U1_SINR}  which can be obtained as
\begin{align}\label{Proof_starts}
   P^{\text{Out}}_{1}= 
\underbrace{\mathrm{Pr}\left(\gamma_{2,1}<\gamma_{2,1}^{th}\right)}_{\text{$P^{\text{Out}}_{2,1}$}}  \text{and}  
\underbrace{\mathrm{Pr}\left(\gamma_{1}<\gamma_{1}^{th}\right)}_{\text{$P^{\text{Out}}_{11}$}}
\end{align}
We will first calculate for $P^{\text{Out}}_{2,1}$. By substituting \eqref{U2_SINR_21} and performing some algebraic manipulation, we obtain 
\begin{align}\label{Proof_U21}
    P^{\text{Out}}_{2,1}= \mathrm{Pr}\left(|\hat{h}_2|^2< \gamma^{th}_{2,1}\left(\frac{|\hat{h}_1|^2\rho_1+A_{2,1}}{\rho_2\left(\alpha_{2,1}-\gamma^{th}_{2,1}\alpha_{2,2} \right)}\right)\!\!\right),
\end{align}
where $A_{2,1} = \left(\rho_1\Omega_{h_{1e}}\!+\!\rho_2\Omega_{h_{2e}}\!\!+\!1\right)$.  Let $X=|\hat{h}_2|^2$\; and $Y=|\hat{h}_1|^2$, then \eqref{Proof_U21} can be expressed as  
\begin{align}\label{Proof_21_Step_1}
             P_{2,1}^{\text{Out}}= \int\limits^{\infty}_{0}f_{Y}(y) F_{X}\left(\frac{\gamma_{2,1}^{th}\left(y\rho_1+A_{2,1}\right)}{\rho_2\left(\alpha_{2,1}-\gamma^{th}_{2,1}\alpha_{2,2} \right)}\right)dy,
         \end{align}
Next, by substituting \eqref{CDF} in \eqref{Proof_21_Step_1}, we obtain 
\begin{align}\label{Proof_21_Step_22}
    P_{2,1}^{\text{Out}}\!\!=\!\! \int\limits^{\infty}_{0}\!\!f_{Y}(y)\frac{1}{\Gamma(m_2)}\gamma\left(\!\!m_2, \frac{m_2\gamma_{2,1}^{th}\left(y B_{2,1}+A_{2,1}\right)}{\hat{\Omega}_2\rho_2\left(\alpha_{2,1}-\gamma^{th}_{2,1}\alpha_{2,2}\right) }\right)dy,
\end{align}
Since, $m_2$ is a positive integer, therefore, this can be changed to Erlang distribution resulting in \eqref{Proof_21_Step_22} expansion as 
\begin{align}\label{Erlang}
    P_{2,1}^{\text{Out}}= &\int\limits^{\infty}_{0}f_{Y}(y)\big[1-\exp\left(-T_{2,1}\left(y B_{2,1}+A_{2,1}\right)\right)\nonumber\\
    &\times\sum\limits_{p=0}^{{m_2}-1}\frac{\left(T_{2,1}\left(y B_{2,1}+A_{2,1}\right)\right)^p}{p!}\big]dy,
\end{align}
where $T_{2,1}\!\!=\!\! \frac{{m_2} \gamma_{2,1}^{th}}{\hat{\Omega}_{2}\rho_2\left(\alpha_{2,1}-\gamma_{2,1}^{th}\alpha_{2,2}\right)}$, $B_{2,1}=\rho_1$. Next, substituting the PDF from \eqref{pdf} in \eqref{Erlang} and applying \cite[ $1.111$]{2015249} on ${\left(T_{2,1}\left(y B_{2,1}+A_{2,1}\right)\right)^p}$, we obtain
\begin{align}\label{Final_expression_U_21}
     P_{2,1}^{\text{Out}}&= 1-\sum\limits_{p=0}^{{m_2}-1}\sum\limits_{q=0}^{p}\binom{p}{q}\left(T_{2,1}B_{2,1}\right)^q  \left(T_{2,1}A_{2,1}\right)^{(p-q)}\nonumber\\
     &\times\frac{\exp\left(-T_{2,1}A_{2,1}\right) C_1^{m_1} }{(p!)\Gamma(m_2)}  \int\limits^{\infty}_{0} y^{v-1}\exp\left(-y\mu\right)dy
\end{align}  
        where $v=(m_1+q)$  and $\mu=(C_1+T_{2,1}B_{2,1})$. The resulting integral can be solved by applying \cite[$3.381.4$]{2015249}. Thus, we obtain \eqref{P_New_out21}. Next, we solve $P_{11}^{\text{Out}}$ by substituting \eqref{U1_SINR} in \eqref{Proof_starts} as 
        \begin{align}\label{Proof_U1_starts}
           P_{11}^{\text{Out}}\! =\!  \mathrm{Pr}\!\!\left(\!|\hat{h}_1|^2\!\!<\! \gamma^{th}_{1}\!\!\left(\!\!\frac{|\hat{h}_2|^2\rho_2\left(\!\alpha_{2,2}\!+\!\Xi_2\alpha_{2,1}\!\right)\!+\!A_1\!\!}{\rho_1}\!\right)\!\!\right),
        \end{align}
        where $A_{1} \!\!=\!\! \left(\rho_1\Omega_{h_{1e}}+\rho_2\Omega_{h_{2e}}\left(\alpha_{2,2}+\alpha_{2,1}\Xi_2\right)+1\right)$.
        Let $X=|\hat{h}_1|^2$\; and $Y=|\hat{h}_2|^2$, then \eqref{Proof_U1_starts} can be expressed as 
        \begin{align}\label{Proof_step11}
             P_{11}^{\text{Out}}= \int\limits^{\infty}_{0}f_{Y}(y) F_{X}\left(\frac{\gamma_{1}^{th}\left(y D_1+A_1\right)}{\rho_1}\right)dy,
         \end{align}
         where $D_1\! \!=\!\! \rho_2\left(\alpha_{2,2}+\Xi_2\alpha_{2,1}\right) $. We then solve \eqref{Proof_step11} by following a similar approach as of \eqref{Proof_21_Step_22} to \eqref{Final_expression_U_21}, and thus obtain \eqref{P_out_11}. 
 \end{proof}
 
  \section{}
 \begin{proof}
   In order to evaluate the  asymptotic outage probability for $U_1$, we first need to follow the same approach as outlined in \eqref{Proof_starts} and \eqref{Proof_U21}.  As $\rho_1=\rho_2\to\infty$, the term $A_{2,1}$ is reduced to 1.
By applying \textit{L'Hôpital's rule}, the resulting limit results in $\frac{\gamma_{2,1}^{th}|\hat{h}_1|^2 }{\alpha_{2,1} - \gamma_{2,1}^{\text{th}} \alpha_{2,2}}
$. Next, following a similar approach till \eqref{Erlang}, the resulting equation can be expressed as 
   \begin{align}\label{Erlang_asymptotic}
   \mathcal{P}^O_{2,1}\!\!= \!\!&\int\limits^{\infty}_{0}\!\!f_{Y}(y)\big[1\!-\!\exp\left(-\mathbb{T}_{2,1}y\right)\sum\limits_{p=0}^{{m_2}-1}\!\!\frac{\left(\mathbb{T}_{2,1
}y\right)^p}{p!}\big]dy,
\end{align}
where $\mathbb{T}_{2,1}\!\!=\!\! \frac{{m_2} \gamma_{2,1}^{th}}{\hat{\Omega}_{2}\left(\alpha_{2,1}-\gamma_{2,1}^{th}\alpha_{2,2}\right)}$. Next, by performing some mathematical simplifications, we obtain $\mathcal{P}^O_{2,1}$ given in \eqref{Asymp_11}. Similarly, we then solve for $\mathcal{P}^O_{11}$ and thus, obtain final expression given in \eqref{Asymp_11}.
   
 \end{proof}
 \section{}
 \begin{proof}
    For evaluating the  asymptotic outage probability of $U_2$  considering  perfect SIC, i.e, $\Xi_1=\Xi_2=0$, we can write the expression as 
    \begin{align}\label{Proof_U22_asymp}
   \mathcal{P}^{O}_{2,2}= \mathrm{Pr}\left(|\hat{h}_2|^2< \frac{\gamma^{th}_{2,2}A_{2,2}}{\rho_2\alpha_{2,2}}\right),
\end{align}
where $A_{2,2} = \left(\rho_1\Omega_{h_{1e}}\Xi_1 + \rho_2\Omega_{h_{2e}}\left(\alpha_{2,2}+\alpha_{2,1}\Xi_2\right)+1 \right)$ . As $\rho_1=\rho_2\to\infty$ and $\Xi_1=\Xi_2=0$, the term $A_{2,2}$ is reduced to $1$. Next,  we use $\gamma\left(c,z\right) \underset{z \to 0}{\approx} \frac{z^c}{c}$, thus, we obtain
    \begin{align}\label{Aym_proof}
       \mathcal{P}_{2,2}^O=\frac{1}{ m_2\;\Gamma(m_2)}\left(\frac{\gamma^{th}_{2,2}m_2}{\hat{\Omega}_2\rho_2\alpha_{2,2}}\right)^{m_2}.
    \end{align}
 \end{proof}


\end{document}